\documentclass{fundam}

\usepackage{url} 
\usepackage{graphicx}
\usepackage{tikz}
\usetikzlibrary{automata, positioning, arrows}
\usepackage{booktabs} 
\usepackage{microtype}
\usepackage{amssymb}
\usepackage{microtype}
\usepackage{times}
\usepackage{algorithmic}
\usepackage{algorithm}
\usepackage{epsfig}
\usepackage{amssymb}
\usepackage{amsmath}
\usepackage{amsfonts}
\usepackage{subfigure}
\usepackage{multirow}
\usepackage{wrapfig}
\usepackage{subfigure}
\usepackage{color}
\usepackage{footnote}
\usepackage{rotating}
\usepackage{pdflscape}

\usepackage{hyperref}

\begin{document}


\setcounter{page}{219}
\publyear{2021}
\papernumber{2071}
\volume{182}
\issue{3}

 \finalVersionForARXIV


\title{Exact and Approximate Algorithms for \\  Computing Betweenness Centrality in Directed Graphs}

\author{Mostafa Haghir Chehreghani\thanks{Address  for correspondence: Department of Computer
      Engineering, Amirkabir University of Technology (Tehran Polytechnic), No. 350, Hafez Ave., Valiasr Square,
            Tehran, Iran.\newline \newline
          \vspace*{-6mm}{\scriptsize{Received  September 2020; \ revised September 2021.}}}
          \\
Department of Computer Engineering \\
Amirkabir University of Technology (Tehran Polytechnic), Iran \\
mostafa.chehreghani@aut.ac.ir
\and Albert Bifet, Talel Abdessalem\\
LTCI, T\'el\'ecom-Paris \\
IP-Paris, France \\
\{albert.bifet, talel.abdessalem\}@telecom-paristech.fr
}

\maketitle

\runninghead{M.H. Chehreghani et al.}{Computing Betweenness Centrality in Directed Graphs}

\vspace*{-5mm}
\begin{abstract}
Graphs (networks) are an important tool to model data in different domains.
Real-world graphs are usually {\em directed},
where the edges have a direction and they are not symmetric.
{\em Betweenness centrality} is an important index
widely used to analyze networks.
In this paper,
first given a directed network $G$ and a vertex $r \in V(G)$,
we propose an exact algorithm to compute betweenness score of $r$.
Our algorithm pre-computes a set $\mathcal{RV}(r)$,
which is used to prune a huge amount of computations that do not contribute to the betweenness score of $r$.
Time complexity of our algorithm depends on $|\mathcal{RV}(r)|$ and it is
respectively
$\Theta(|\mathcal{RV}(r)|\cdot|E(G)|)$ and
$\Theta(|\mathcal{RV}(r)|\cdot|E(G)|+|\mathcal{RV}(r)|\cdot|V(G)|\log |V(G)|)$
for unweighted graphs and weighted graphs with positive weights.
$|\mathcal{RV}(r)|$ is bounded from above by $|V(G)|-1$ and in most cases,
it is a small constant.
Then, for the cases where $\mathcal{RV}(r)$ is large,
we present a simple randomized algorithm
that samples from $\mathcal{RV}(r)$ and performs computations for only
the sampled elements.
We show that this algorithm provides an $(\epsilon,\delta)$-approximation to the betweenness score of $r$.
Finally, we perform extensive experiments over several real-world datasets from different domains
for several randomly chosen vertices as well as for the vertices
with the highest betweenness scores.
Our experiments reveal that for estimating betweenness score of a single vertex,
our algorithm significantly outperforms
the most efficient existing randomized algorithms, in terms of both running time and accuracy.
Our experiments also reveal that
our algorithm
improves the existing algorithms when someone
is interested in computing betweenness values of the vertices in a set
whose cardinality is very small.
\end{abstract}


\section{Introduction}
\label{sec:introduction}

Graphs (networks) provide an important tool to model data in different domains,
including social networks, bioinformatics, road networks, the world wide web and communication systems.
A property seen in most of these real-world networks
is that the links between vertices do not always represent reciprocal relations \cite{Newman03thestructure}.
As a result, the networks formed in these domains are {\em directed graphs}
where any edge has a direction and the edges are not always symmetric.

{\em Centrality} is a structural property of vertices (or edges) in the network
that quantifies their relative importance.
For example, it determines the importance of a person within a social network,
or a road within a road network.
Freeman \cite{jrnl:Freeman} introduced and defined
{\em betweenness centrality} of a vertex
as the number of shortest paths
from all (source) vertices to all others that pass through that vertex.
He used it for measuring the control of a human over
the communications among others in a social network \cite{jrnl:Freeman}.
Betweenness centrality is also used in some well-known algorithms for
clustering and community detection in social and information networks \cite{jrnl:Girvan}.

Although there exist polynomial time and space algorithms for betweenness centrality computation,
the algorithms are expensive in practice.
Currently, the most efficient existing exact method is Brandes's algorithm \cite{jrnl:Brandes}.
Time complexity of this algorithm is $\Theta(|V(G)|\cdot|E(G)|)$ for unweighted graphs and
$\Theta(|V(G)|\cdot|E(G)| + |V(G)|^2 \log |V(G)|)$ for weighted graphs with positive weights.
This means this algorithm is not
applicable, even for mid-size networks.

However, there are observations that may improve computation of betweenness centrality
in practice.
In several applications it is sufficient to
compute betweenness score of only one or a few vertices.
For instance,
the
index might be computed only for core vertices of communities in social/information
networks \cite{jrnl:Wang}
or only for hubs in communication networks.
Another example, discussed in \cite{DBLP:conf/complenet/AgarwalSCI15,DBLP:journals/corr/AgarwalSCI14},
is handling cascading failures.
It has been shown that the failure of a vertex with a higher
betweenness score usually causes
a greater collapse of the network
\cite{Stergiopoulos201534}.
Therefore, failed vertices should be recovered in the order of their betweenness scores.
This means it is required to compute betweenness scores of only failed vertices,
that usually form a very small subset of all vertices.
Another example,
discussed in \cite{DBLP:journals/cj/Chehreghani14},
is a road network wherein it is required to compute
betweenness score of a single vertex (intersection)
in different configurations, to see which one is better
in reducing the traffic jam of the intersection.
The other example is in a transportation network.
It is shown that in a transportation network, betweenness
centrality is positively related to the efficiency
of an airport \cite{RePEc:pra:mprapa:17673}.
Hence and as suggested in \cite{DBLP:journals/jea/BergaminiCDMSV18},
when betweenness score of a given (specific) airport node is not large enough,
it should be increased by e.g., adding new edges to the network.
To do so, we need to quickly and precisely estimate
betweenness score of the airport node.
Note that in these applications,
the target vertices are not necessarily
those that have the highest betweenness scores.
Hence, algorithms that identify vertices
with the highest betweenness scores \cite{Riondato2016} are not applicable.
Note also that
it is a famous conjuncture in graph theory whether
betweenness centrality of a single vertex can be computed more efficient than all vertices.

In the current paper, we exploit this observation to design more effective exact and approximate algorithms
for computing betweenness centrality
of a single node or a small set of nodes in a large directed graph.
Our algorithms are based on computing the set of {\em reachable vertices} for a given vertex $r$.
On the one hand, this set can be computed very efficiently.
On the other hand, it indicates the potential source vertices whose contributions (dependency scores) on $r$ are non-zero.
As a result, it helps us to avoid a huge amount of computations that do not contribute to the
betweenness score of $r$.

\medskip
In this paper, our key contributions are as follows.
\begin{itemize}
\item
Given a directed graph $G$ and a vertex $r \in V(G)$,
we present an efficient exact algorithm to compute betweenness score of $r$.
The algorithm is based on pre-computing the set of {\em reachable vertices} of $r$,
denoted by $\mathcal{RV}(r)$.
$\mathcal{RV}(r)$ can be computed in $\Theta(|E(G)|)$ times for both
unweighted graphs and weighted graphs with positive weights.
Time complexity of the whole exact algorithm depends on the size of $\mathcal{RV}(r)$
and it is respectively
$\Theta(|\mathcal{RV}(r)|\cdot|E(G)|)$ and
$\Theta(|\mathcal{RV}(r)|\cdot|E(G)|+|\mathcal{RV}(r)|\cdot|V(G)|\log |V(G)|)$
for unweighted graphs and weighted graphs with positive weights.
$|\mathcal{RV}(r)|$ is bounded from above by $|V(G)|$ and in most cases,
it can be considered as a small constant (see Section~\ref{sec:experimentalresults}).
Hence, in many cases, time complexity of our proposed exact algorithm for unweighted graphs
is linear, in terms of $|E(G)|$, and it is $\Theta(|E(G)|+|V(G)|\log |V(G)|)$
for weighted graphs with positive weights.

\item
In the cases where $\mathcal{RV}(r)$ is large,
our exact algorithm might be intractable in practice.
To address this issue, we present a simple randomized algorithm
that samples elements from $\mathcal{RV}(r)$ and performs computations for only
the sampled elements.
We show that this algorithm provides an $(\epsilon,\delta)$-approximation to the betweenness score of $r$.

\item
In order to evaluate the empirical efficiency of our proposed algorithms,
we perform extensive experiments over several real-world datasets from different domains.
In our experiments, we introduce a procedure that first
computes $\mathcal{RV}(r)$.
Then if the size of $\mathcal{RV}(r)$
is less than some threshold (e.g., $1000$),
it employs the exact algorithm.
Otherwise, it exploits the randomized algorithm.
We evaluate this procedure for several randomly chosen vertices as well as for the vertices
with the highest betweenness scores.
We show that for randomly chosen vertices,
our proposed procedure always significantly outperforms
the most efficient existing randomized algorithms, in terms of both running time and accuracy.
Furthermore, for the vertices that have the highest betweenness scores,
over most of the datasets our algorithm outperforms
most efficient existing algorithms.

\item
While our algorithm is intuitively designed to estimate
betweenness score of only one vertex,
in our experiments we consider the cases wherein betweenness scores of small sets of vertices are computed.
Our experiments reveal that in such cases,
our proposed algorithm efficiently computes betweenness scores of
all vertices in  sets of sizes 5, 10 and 15 and it considerably outperforms the existing algorithms.
\end{itemize}

A preliminary version of this paper
was presented in {\em Proceedings of the
22nd Pacific-Asia Conference on Knowledge Discovery and Data Mining} (PAKDD 2018), pp. 752-764 \cite{DBLP:journals/corr/abs-1708-08739}.
The current paper extends it by a full elaboration
of proofs and theoretical discussions, as well as
a significantly more extensive experimental evaluation.

The rest of this paper is organized as follows.
In Section~\ref{sec:preliminaries},
preliminaries and necessary definitions related to betweenness centrality are introduced.
A brief overview on related work is given in Section~\ref{sec:relatedwork}.
In Section~\ref{sec:betweennessdirected}, we
present our exact and approximate algorithms and their analysis.
In Section~\ref{sec:experimentalresults}, we empirically evaluate our proposed algorithm
and show its high efficiency and accuracy, compared to existing algorithms.
Finally, the paper is concluded in Section~\ref{sec:conclusion}.

\section{Preliminaries}
\label{sec:preliminaries}

In this section, we present definitions and notations widely used in the paper.
We assume that the reader is familiar with basic concepts in graph theory.
Throughout the paper, $G$ refers to a graph (network).
For simplicity, we assume that $G$ is a directed, connected and loop-free graph without multi-edges.
Throughout the paper, we assume that $G$ is an unweighted graph,
unless it is explicitly mentioned that $G$ is weighted.
$V(G)$ and $E(G)$ refer to the set of vertices and the set of edges of $G$, respectively.
For a vertex $v \in V(G)$, the number of head ends adjacent to $v$
is called its {\em in degree}, and
the number of tail ends adjacent to $v$ is called its {\em out degree}.

\medskip
A \textit{shortest path}
from $u \in V(G)$ to $v \in V(G)$ is a path
whose length is minimum, among all paths from $u$ to $v$.
For two vertices $u,v \in V(G)$, if $G$ is unweighted,
by $d(u,v)$ we denote the length (the number of edges) of a shortest path connecting $u$ to $v$.
If $G$ is weighted,
$d(u,v)$ denotes the sum of the weights of the edges of a shortest path connecting $u$ to $v$.
By definition, $d(u,u)=0$.
Note that in directed graphs, $d(u,v)$ is not necessarily equal to $d(v,u)$.
For $s,t \in V(G)$, $\sigma_{st}$ denotes the number of shortest paths between $s$ and $t$, and
$\sigma_{st}(v)$ denotes the number of shortest paths between $s$ and $t$ that also pass through $v$.
{\em Betweenness centrality} of a vertex $v$ is defined as:
\begin{equation}
BC(v)= \sum_{s,t \in V(G) \setminus \{v\}} \frac{\sigma_{st}(v)}{\sigma_{st}}.
\end{equation}

A notion which is widely used for counting the number of shortest paths in a graph is the directed acyclic graph (DAG)
containing all shortest paths starting from a vertex $s$ (see e.g., \cite{jrnl:Brandes}).
In this paper, we refer to it as the \textit{shortest-path-DAG}, or \textit{SPD} in short, rooted at $s$.
For every vertex $s$ in graph $G$,
the \textit{SPD} rooted at $s$ is unique,
and it can be computed in $\Theta(|E(G)|)$ time for unweighted graphs
and in $\Theta\left(|E(G)|+|V(G)|\text{ log }|V(G)|\right)$ time
for weighted graphs with positive weights \cite{jrnl:Brandes}.

\medskip
Brandes \cite{jrnl:Brandes} introduced the notion of the \textit{dependency score}
of a vertex $s \in V(G)$ on a vertex $v \in V(G) \setminus \{s\}$, which is defined as:
\begin{equation}
\delta_{s\bullet}(v)=\sum_{t \in V(G) \setminus \{v,s\}} \delta_{st}(v)
\end{equation}
where
$\delta_{st}(v) = \frac {\sigma_{st}(v)}{\sigma_{st}}.$
We have:
\begin{equation}
BC(v)= \sum_{s \in V(G) \setminus \{v\}} \delta_{s\bullet}(v).
\end{equation}

Brandes \cite{jrnl:Brandes} showed that dependency scores of a
source vertex on different vertices in the network can be computed using a recursive relation,
defined as the following:
\begin{equation}
\delta_{s\bullet}(v)=\sum_{w:v \in P_s(w)} \frac{\sigma_{sv}}{\sigma_{sw}}(1+\delta_{s\bullet}(w)),
\label{eq:recursive}
\end{equation}
where
$P_s(w)$ contains the predecessors of $w$ in the SPD rooted at $s$.


 \section{Related work}
\label{sec:relatedwork}

Brandes~\cite{jrnl:Brandes} introduced an efficient algorithm
for computing betweenness centrality of a vertex,
which is performed in
$\Theta( |V(G)| |E(G)| )$ and $\Theta(|V(G)| |E(G)| + |V(G)|^2 \log |V(G)|)$
times for unweighted and weighted networks with positive weights, respectively.
{\c{C}}ataly{\"{u}}rek et.al. \cite{DBLP:conf/sdm/CatalyurekKSS13} presented the
{\em compression} and {\em shattering} techniques to improve
the efficiency of Brandes's algorithm
for large graphs. During {\em compression}, vertices with known betweenness scores are removed from the graph and
during {\em shattering}, the graph is partitioned into smaller components.
Holme \cite{jrnl:Holme} showed that betweenness centrality of a vertex is
highly correlated with the fraction of time that the vertex is occupied
by the traffic of the network.
Barthelemy \cite{jrnl:Barthelemy} showed that many scale-free networks \cite{jrnl:Barabasi}
have a power-law distribution of betweenness centrality.
Furno et.al.~\cite{DBLP:conf/bigdataconf/FurnoFSZ17}
reduced the number of shortest-path-DAGs by using, as sources, pivot nodes identified through the exploitation of topological properties of graphs revealed by using clustering.
They empirically evaluated their algorithm
over a real-world road network and showed that
the approximation error does not significantly affect the most critical vertices.

\subsection{Generalization to sets}
Everett and Borgatti \cite{jrnl:Everett} defined \textit{group betweenness centrality}
as a natural extension of betweenness centrality for sets of vertices.
Group betweenness centrality of a set is defined as the
number of shortest paths passing through at least one of the vertices in the set \cite{jrnl:Everett}.
The other natural extension of betweenness centrality is \textit{co-betweenness centrality}.
Co-betweenness centrality is defined as the number of shortest paths passing through all vertices in the set.
Kolaczyk et.al. \cite{jrnl:Kolaczyk} presented an $\Theta(|V(G)|^3)$ time algorithm for
co-betweenness centrality computation of sets of size 2.
Chehreghani \cite{conf:cbcwsdm} presented efficient algorithms for co-betweenness centrality computation
of any set or sequence of vertices in weighted and unweighted graphs.
Puzis et.al. \cite{jrnl:PuzisPhysRev} proposed an $\Theta(|K|^3)$ time algorithm for
computing successive group betweenness centrality,
where $|K|$ is the size of the set.
The same authors in \cite{jrnl:PuzisAIComm} presented two algorithms
for finding \textit{most prominent group}.
A \textit{most prominent group} of a network is a set of vertices of minimum size,
so that every shortest path in the network passes through at least one of the vertices in the set.
The first algorithm is based on a heuristic search and
the second one is based on iterative greedy choice of vertices.
Chehreghani et.al.~\cite{DBLP:conf/bigdataconf/ChehreghaniBA18} compared different sampling algorithms for estimating group betweenness centrality.
More than the standard techniques presented in the literature, they investigated a method which is based on the distance between a single vertex and a set of vertices.

\subsection{Approximate algorithms}

Brandes and Pich
\cite{jrnl:Brandes3} proposed an approximate algorithm based on
selecting $k$ source vertices and
computing dependency scores of them on the other vertices in the graph.
They used various strategies for selecting the source vertices, including:
MaxMin, MaxSum and MinSum \cite{jrnl:Brandes3}.
In the method of \cite{proc:Bader}, some source vertices are selected uniformly at random,
and their dependency scores are computed and scaled for all vertices.
Geisberger et.al. \cite{conf:Geisberger} presented an algorithm
for approximate ranking of vertices based on their betweenness scores.
In this algorithm, the method for aggregating dependency
scores changes so that vertices do not profit from being near the selected source vertices.
Chehreghani~\cite{DBLP:journals/cj/Chehreghani14}
proposed a randomized framework for
unbiased estimation of the betweenness score of a single vertex.
Then,
to estimate betweenness score of vertex $v$,
he proposed a non-uniform sampler, defined as follows:
\[\mathbb P[s] = \frac{\frac{1}{d(v,s)}}{\sum_{u \in V(G)\setminus \{v\}} \frac{1}{d(v,u)}},\]
where $s \in V(G) \setminus \{v\}$.

\medskip
Riondato and Kornaropoulos \cite{Riondato2016}
presented shortest path samplers for estimating betweenness centrality of
all vertices or the $k$ vertices
that have the highest betweenness scores in a graph.
They determined the number of samples
needed to approximate the betweenness with the desired accuracy
and confidence by means of the VC-dimension theory \cite{vc-ucrfep-71}.
Recently, Riondato and Upfal \cite{RiondatoKDD20116} introduced algorithms
for estimating betweenness scores of all vertices in a graph.
They also discussed a variant of the algorithm that finds the top-$k$ vertices.
They used Rademacher average \cite{Shalev-Shwartz:2014:UML:2621980} to determine the number
of required samples.
Borassi and Natale \cite{DBLP:conf/esa/BorassiN16} presented the KADABRA algorithm,
which uses balanced bidirectional BFS (bb-BFS) to sample shortest paths.
In bb-BFS, a BFS is performed from each of the two endpoints $s$ and $t$,
in such a way that they
explore
almost the same number of edges.
The authors of~\cite{DBLP:conf/edbt/ChehreghaniAB19} investigated using the Metropolis-Hastings technique
to sample from the optimal distribution
presented in \cite{DBLP:journals/cj/Chehreghani14} for betweenness centrality estimation.

\subsection{Dynamic graphs}
Lee et.al.~\cite{proc:www}
proposed an algorithm to efficiently update betweenness centrality of
vertices when the graph obtains a new edge.
They reduced the search space by finding a candidate set of
vertices whose betweenness scores
can be updated.
Bergamini et.al. \cite{DBLP:conf/alenex/BergaminiMS15} presented approximate algorithms
that update betweenness scores
of all vertices when an edge is inserted or
the weight of an edge decreases.
They used the algorithm of \cite{Riondato2016} as the building block.
Hayashi et.al. \cite{DBLP:journals/pvldb/HayashiAY15}
proposed a fully dynamic algorithm for estimating betweenness centrality of all vertices in a large
dynamic network.
Their algorithm is based on two data structures: {\em hypergraph sketch} that
keeps track of SPDs,
and {\em two-ball index} that helps
to identify the parts of hypergraph sketches that require updates.
An overview on dynamical algorithms for updating  betweenness centrality in dynamic graphs can be found in \cite{https://doi.org/10.1002/widm.1393}.

\section{Computing betweenness centrality in directed graphs}
\label{sec:betweennessdirected}

In this section,
we present our exact and approximate algorithms
for computing betweenness centrality of a given vertex $v$ in a large directed graph.
First in Section~\ref{sec:reachable}, we introduce {\em reachable vertices} and show that they are sufficient
to compute the betweenness score of $v$.
Then in Sections~\ref{sec:exact} and \ref{sec:inexact},
we respectively present our exact and approximate algorithms.

\subsection{Reachable vertices}
\label{sec:reachable}

Let $G$ be a directed graph and $r \in V(G)$.
Suppose that we want to compute betweenness score of $r$.
To do so, as Brandes algorithm \cite{jrnl:Brandes} suggests,
for each vertex $s \in V(G)$, we may form the SPD rooted at $s$ and compute
the dependency score of $s$ on $r$.
Betweenness score of $r$ will be the sum of all the dependency scores.
However, it is possible that in a directed graph and for many vertices $s$,
there is no path from $s$ to $r$ and as a result, dependency score of $s$ on $r$ is 0.
An example of this situation is depicted in Figure~\ref{fig:rv1}.
In the graph of this figure, suppose that we want to compute betweenness score of vertex $r_1$.
If we form the SPD rooted at $v_1$, after visiting the parts of the graph indicated by hachures,
we find out that there is no shortest path from $v_1$ to $r_1$ and hence,
$\delta_{v_1\bullet}(r_1)$ is 0.
The same holds for all vertices in the hachured part of the graph, i.e.,
dependency scores of these vertices on $r_1$ are 0.
The question arising here is that whether there exists an efficient way to detect
the vertices whose dependency scores on $r$ are 0
(so that we can avoid forming SPDs rooted at them)?
In the rest of this section,
we aim to answer this question.
We first introduce a usually small subset of vertices,
called {\em reachable vertices} and denoted with $RV(r)$,
that are sufficient to compute betweenness score of~$r$.
Then, we discuss how this set can be computed efficiently.

\begin{figure}[h]
\centering
\subfigure[]
{
\includegraphics[scale=0.42]{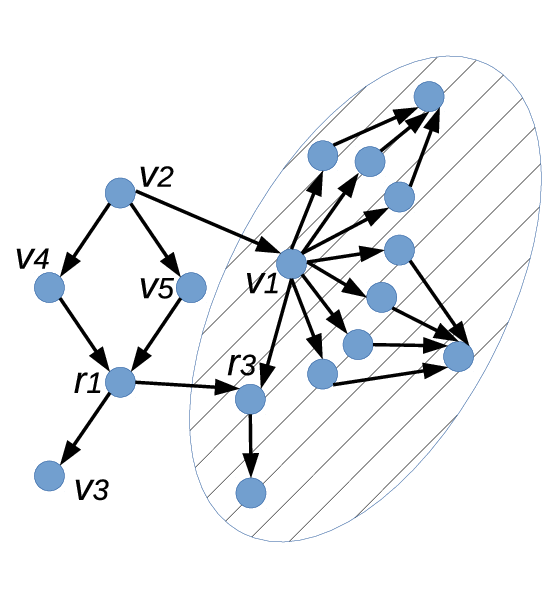}
\label{fig:rv1}
}\qquad\quad
\subfigure[]
{
\includegraphics[scale=0.42]{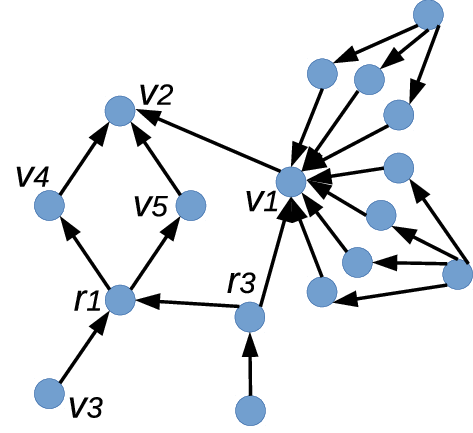}
\label{fig:rv2}
}
\subfigure[]
{
\hspace*{5mm}\includegraphics[scale=0.42]{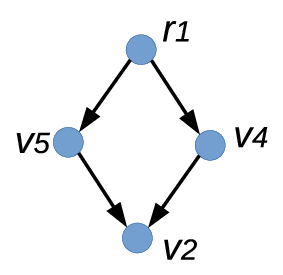}
\label{fig:rv3}
}\vspace*{-2mm}
\caption
{
\label{fig:rv}
In Figure~\ref{fig:rv1},
the dependency scores of the vertices in the hachured part of the graph (and also $v_3$) on $r_1$ is 0.
The graph of Figure~\ref{fig:rv2} presents the reverse graph of the graph of Figure~\ref{fig:rv1}.
Figure~\ref{fig:rv3} shows how $\mathcal{RV}(r_1)$ is computed.
}
\end{figure}

\begin{definition}
Let $G$ be a directed graph and $r,v \in V(G)$.
We say $r$ is {\em reachable} from $v$ if there is a (directed) path from $v$ to $r$.
The set of vertices that $r$ is reachable from them is denoted by $RV(r)$.
\end{definition}

\begin{proposition}
\label{proposition:reachable}
Let $G$ be a directed graph and $r \in V(G)$.
If {\em out degree} of $r$ is $0$, $BC(r)$ is $0$, too.
Otherwise, we have:
\begin{equation}
\label{eq:rv}
BC(r)=\sum_{v\in RV(r)} \delta_{v\bullet}(r).
\end{equation}
\end{proposition}
\begin{proof}
If {\em out degree} of $r$ is $0$, there is no shortest path in the graph that leaves $r$, as a result,
$BC(r)$ is $0$.
To prove that Equation~\ref{eq:rv} holds, we need to prove that for any $w\in V(G)\setminus RV(r)$,
dependency score of $w$ on $r$ is $0$.
Obviously, this holds, because there is no path from $w$ to $r$ and as a result,
no shortest path starting from $w$ can pass over $r$.
\end{proof}

Proposition~\ref{proposition:reachable} suggests that for computing betweenness score of $r$,
we first check whether {\em out degree} of $r$ is greater than $0$ and if so,
we compute $RV(r)$.
Betweenness score of $r$ is exactly computed
using Equation~\ref{eq:rv}.

\medskip
If $RV(r)$ is already known, this procedure can significantly improve
computation of betweenness centrality of $r$.
The reason is that, as our experiments show,
in real-world directed networks $RV(r)$ is usually significantly smaller than $V(G)$.
However, computing $RV(r)$ can be computationally expensive as in the worst case,
it requires the same amount of time as computing betweenness score of $r$.
This motivates us to try to define a set $\mathcal{RV}(r)$ that satisfies the following properties:
(i) $ RV(r) \subseteq \mathcal{RV}(r)$ and
(ii) $\mathcal{RV}(r)$ can be computed effectively in a time much faster than computing $BC(r)$.
Condition (i) implies that each vertex $v \in V(G)$ whose dependency score on $r$ is greater than $0$,
belongs to $\mathcal{RV}(r)$ and as a result,
$BC(r) = \sum_{v \in \mathcal{RV}(r)} \delta_{v\bullet}(r).$
In the following,
we present a definition of $\mathcal{RV}(r)$ and a simple and efficient algorithm to compute it.

\begin{definition}
Let $G$ be a directed graph.
{\em Reverse graph} of $G$, denoted by $R(G)$,
is a directed graph such that:
(i) $V(R(G))=V(G)$, and
(ii) $(u,v) \in E(R(G))$ if and only if $(v,u) \in E(G)$.
\end{definition}

For example,
the graph of Figure~\ref{fig:rv2} presents the {\em reverse graph} of the graph of Figure~\ref{fig:rv1}.

\begin{definition}
\label{def:mathcalrv}
Let $G$ be a directed graph and $r \in V(G)$.
We define $\mathcal{RV}(G)$ as the set that contains any vertex $v$
such that there is a path from $r$ to $v$ in $R(G)$.
\end{definition}

\begin{proposition}
\label{proposition:rv}
Let $G$ be a directed graph and $r \in V(G)$.
We have: $RV(r) = \mathcal{RV}(r)$.
\end{proposition}
\begin{proof}
The proof is straight-forward from the definitions of $RV(r)$ and $\mathcal{RV}(r)$.
For each $v\in V(G)$,
if $v \in RV(r)$, then there is a path from $v$ to $r$ and as a result,
there is a path from $r$ to $v$ in $R(G)$.
Hence, $v \in \mathcal{RV}(r)$ and therefore,
$RV(r) \subseteq \mathcal{RV}(r)$.
In a similar way, we can show that $\mathcal{RV}(r) \subseteq RV(r)$.
Therefore, we have: $RV(r) = \mathcal{RV}(r)$.
\end{proof}

An advantage of the above definition of
$\mathcal{RV}(r)$
is that it can be efficiently computed as follows:
\begin{enumerate}
\item first, by flipping the direction of the edges of $G$,
$R(G)$ is constructed.
\item then, if $G$ is weighted, the weights of the edges are ignored,
\item finally, a breadth first search (BFS) or a depth-first search (DFS) on $R(G)$ starting from $r$ is performed.
All the vertices that are met during the BFS (or DFS), except $r$, are added to $\mathcal{RV}(r)$.
\end{enumerate}

In fact, while in $RV(r)$ we require to solve the multi-source shortest path problem (MSSP),
in $\mathcal{RV}(r)$ this is reduced to the single-source shortest path problem (SSSP),
which can be addressed much faster.
Figure~\ref{fig:rv} shows an example of this procedure,
where in order to compute $\mathcal{RV}(r_1)$,
we first generate $R(G)$ (Figure~\ref{fig:rv2})
and then, we run a BFS (or DFS) starting from $r_1$ (Figure~\ref{fig:rv3}).
The set of vertices that are met during the traversal except $r_1$,
i.e., vertices $v_2$, $v_4$ and $v_5$,
form $\mathcal{RV}(r_1)$.

For a vertex $r \in V(G)$,
each of the steps of the procedure of computing $\mathcal{RV}(r)$,
for both unweighted graphs and weighted graphs,
can be computed in $\Theta(|E(G)|)$ time.
Hence, time complexity of the procedure of computing $\mathcal{RV}(r)$
for both unweighted graphs and weighted graphs
is $\Theta(|E(G)|)$.
Therefore, $\mathcal{RV}(r)$ can be computed in a time much faster than computing betweenness score f $r$.
Furthermore, Proposition~\ref{proposition:rv} says that $\mathcal{RV}(r)$ contains all the members of $RV(r)$.
These two imply that both of
the afore-mentioned conditions are satisfied.

\subsection{The exact algorithm}
\label{sec:exact}

In this section,
using the notions and definitions presented in Section~\ref{sec:reachable},
we propose an effective algorithm to compute exact betweenness score of a given vertex $r$ in
a directed graph $G$.

Algorithm~\ref{algorithm:ex-bcd} presents the high level
pseudo code of the \textsf{E-BCD} algorithm proposed for
computing exact betweenness score of $r$ in $G$.
After checking whether or not {\em out degree} of $r$ is $0$,
the algorithm follows two main steps:
(\textit{i}) computing $\mathcal{RV}(G)$
(Lines~\ref{line:step1_1}-\ref{line:step1_2} of Algorithm~\ref{algorithm:ex-bcd}),
where we use the procedure described in Section~\ref{sec:reachable}
to compute $\mathcal{RV}(r)$; and
(\textit{ii}) computing $BC(r)$
(Lines~\ref{line:step2_1}-\ref{line:step2_2} of Algorithm~\ref{algorithm:ex-bcd}),
where for each vertex $v \in \mathcal{RV}(r)$,
we form the SPD rooted at $v$ and compute the dependency score of $v$ on the other vertices and
add the value of $\delta_{v\bullet}(r)$ to the betweenness score of $r$.
Note that if $G$ is weighted, while in the first step the weights of its edges are ignored,
in the second step
and during forming SPDs and computing dependency scores,
we take the weights into account.

Note also that in Algorithm~\ref{algorithm:ex-bcd},
after computing $\mathcal{RV}(r)$,
techniques proposed to improve exact betweenness centrality computation,
such {\em compression} and {\em shattering} \cite{DBLP:conf/sdm/CatalyurekKSS13},
can be used to improve the efficiency of the second step.
This means the algorithm proposed here
is orthogonal to the techniques such as shattering and compression and therefore,
they can be merged.

\paragraph{Complexity analysis}
On the one hand, as mentioned before, time complexity of the first step is $\Theta(|E(G)|)$.
On the other hand,
time complexity of each iteration in Lines~\ref{line:loop1}-\ref{line:loop2}
is $\Theta(|E(G)|)$ for unweighted graphs
and $\Theta(|E(G)|+|V(G)|\log|V(G)|)$ for weighted graphs with positive weights.
As a result, time complexity of \textsf{E-BCD} is
$\Theta(|\mathcal{RV}(G)|\cdot |E(G)|)$ for unweighted graphs and
$\Theta(|\mathcal{RV}(G)|\cdot |E(G)|+|\mathcal{RV}(G)|\cdot |V(G)|\log|V(G)|)$
for weighted graphs with positive weights.
Since most of vertices in real-world networks have a small
reachable set (see Section~\ref{sec:experimentalresults}),
this time complexity improves time complexity of Brandes' algorithm~\cite{jrnl:Brandes}.

\medskip

\begin{algorithm}[!ht]
\caption{High level pseudo code of the algorithm of
computing exact betweenness centrality in directed graphs.}
\label{algorithm:ex-bcd}
\begin{algorithmic} [1]
\STATE \textsf{E-BCD}
\STATE \textbf{Input.} A directed network $G$ and a vertex $r \in V(G)$.
\STATE \textbf{Output.} Betweenness score of $r$.
\IF{{\em out degree} of $r$ is $0$}
\RETURN $0$.
\ENDIF

\STATE \COMMENT{Compute $\mathcal{RV}(r)$:} \label{line:step1_1}
\STATE $\mathcal{RV}(r) \leftarrow \emptyset$.
\STATE $R(G) \leftarrow$ compute the reverse graph of $G$.
\STATE If $G$ is weighted, ignore the weights of the edges of $R(G)$.
\STATE Perform a BFS or DFS on $R(G)$ starting from $r$.
\STATE Add to $\mathcal{RV}(r)$ all the visited vertices, except $r$. \label{line:step1_2}

\STATE \COMMENT{Compute $BC(r)$:} \label{line:step2_1}
\STATE $bc \leftarrow 0$.
\FORALL{vertices $v\in \mathcal{RV}(G)$} \label{line:loop1}
\STATE Form the SPD rooted at $v$ and compute the dependency scores of $v$ on the other vertices.
\STATE $bc \leftarrow bc+ \delta_{v\bullet}(r)$.
\ENDFOR \label{line:loop2} \label{line:step2_2}
\RETURN $bc$.
\end{algorithmic}
\end{algorithm}

\textsf{E-BCD} can be simply revised to compute
betweenness scores of all vertices in a set
$R=\{r_1,\ldots,r_l\}$ ($l$ is the cardinality of $R$).
Let $\mathcal D$ be $\cup_{r_i \in R} \mathcal{RV}(r_i)$.
After forming the SPD rooted at each vertex in
$\mathcal D$, we can easily compute betweenness scores of all the vertices in $R$.
Someone may wonder for what sizes of $R$
\textsf{E-BCD} yields a better algorithm than
Brandes' algorithm.
Assume that forming the SPD rooted at a vertex $v_i$ and computing dependency scores of $v_i$ on other vertices takes the same time $f$ as finding
$\mathcal{RV}(v_i)$.
While in practice the former takes usually more time than the latter (as it needs to traverse the SPD twice), this assumption can be particularly valid in theory for unweighted graphs, as both of these two operations have the same asymptotic time complexity.
\textsf{E-BCD} spends
$l \cdot f+ |\mathcal D| \cdot f  =
(l+|\mathcal D|) \cdot f$ time to compute the scores.
Brandes' algorithm on the other hand spends
$|V(G)| \cdot f$ time.
This means if $l < |V(G)| - |\mathcal D|$,
\textsf{E-BCD} outperforms Brandes' algorithm,
otherwise, Brandes' algorithm will have a smaller
time complexity.

\subsection{The approximate algorithm}
\label{sec:inexact}

For a vertex $r \in V(G)$,
$\mathcal{RV}(r)$ is always smaller than $|V(G)|$ and
as our experiments (reported in Section~\ref{sec:experimentalresults}) show,
the difference is usually significant.
Therefore, \textsf{E-BCD} is usually significantly more efficient than the existing exact algorithms
such as Brandes's algorithm \cite{jrnl:Brandes}.
However, in some cases, the size of $\mathcal{RV}(r)$
can be large (see again Section~\ref{sec:experimentalresults}).
To make the algorithm tractable for the cases where $\mathcal{RV}(r)$ is large,
in this section we propose a randomized algorithm that picks some elements of $\mathcal{RV}(r)$
uniformly at random and only processes these vertices.

Algorithm~\ref{algorithm:app-bcd} shows the high level pseudo code of our
randomized algorithm, called \textsf{A-BCD}.
Similar to \textsf{E-BCD}, \textsf{A-BCD} first
computes $\mathcal{RV}(r)$.
Then, at each iteration $t$ ($1 \leq t \leq T$),
\textsf{A-BCD} picks a vertex $v$ from $\mathcal{RV}(r)$
uniformly at random,
forms the SPD rooted at $v$ and computes $\delta_{v\bullet}(r)$.
In the end, betweenness of $r$ is estimated as the sum of the computed dependency scores on $r$
multiply by $\frac{|\mathcal{RV}(r)|}{T}$.

\begin{algorithm}[!htb]
\caption{High level pseudo code of the algorithm of
computing approximate betweenness centrality in directed graphs.}
\label{algorithm:app-bcd}
\begin{algorithmic} [1]
\STATE \textsf{A-BCD}
\STATE \textbf{Input.} A network $G$, a vertex $r\in V(G)$ and the number of samples $T$.
\STATE \textbf{Output.} Estimated betweenness score of $r$.
\IF{{\em out degree} of $r$ is $0$}
\RETURN $0$.
\ENDIF

\STATE \COMMENT{Compute $\mathcal{RV}(r)$:}
\STATE $\mathcal{RV}(r) \leftarrow \emptyset$.
\STATE $R(G) \leftarrow$ compute the reverse graph of $G$.
\STATE If $G$ is weighted, ignore the weights of the edges of $R(G)$.
\STATE Perform a BFS or DFS on $R(G)$ starting from $r$.
\STATE Add to $\mathcal{RV}(r)$ all visited vertices, except $r$.

\STATE \COMMENT{Estimate $BC(r)$:}
\STATE $bc \leftarrow 0$.
\FORALL{$t=1$ \textbf{to} $T$ } \label{line:loop3}
\STATE Select a vertex $v_t \in \mathcal{RV}(r)$ uniformly at random.
\STATE Form the SPD rooted at $v_t$ and compute dependency scores of $v_t$ on the other vertices.
\STATE $bc \leftarrow bc+ \frac{\delta_{v_t\bullet}(r) \cdot |\mathcal{RV}(r)|}{T} $.
\ENDFOR \label{line:loop4}
\RETURN $bc$.
\end{algorithmic}
\end{algorithm}

\paragraph{Complexity analysis}
Similar to \textsf{E-BCD},
on the one hand, time complexity of the $\mathcal{RV}(r)$ computation step is $\Theta(|E(G)|)$.
On the other hand,
time complexity of each iteration in Lines~\ref{line:loop3}-\ref{line:loop4} of Algorithm~\ref{algorithm:app-bcd}
is $\Theta(|E(G)|)$ for unweighted graphs
and $\Theta(|E(G)|+|V(G)|\log|V(G)|)$ for weighted graphs with positive weights.
As a result, time complexity of \textsf{A-BCD} is
$\Theta(T\cdot |E(G)|)$ for unweighted graphs and
$\Theta(T\cdot |E(G)|+T\cdot |V(G)|\log|V(G)|)$
for weighted graphs with positive weights,
where $T$ is the number of iterations (samples).

\paragraph{Error bound}
Using Hoeffding's inequality \cite{Hoeffding:1963},
we can simply derive an error bound for the estimated value of
betweenness score of $r$.
First in Proportion~\ref{proposition:expectedvalue},
we prove that in Algorithm~\ref{algorithm:app-bcd} the expected value of $bc$ is $BC(r)$.
Then in Proportion~\ref{proposition:errorbound},
we provide an error bound for $bc$.

\begin{proposition}
\label{proposition:expectedvalue}
In Algorithm~\ref{algorithm:app-bcd}, we have: $\mathbb E\left[bc\right]=BC(r)$.
\end{proposition}
\begin{proof}
For each $t$, $ 1 \leq t \leq T$, we define random variable $bc_t$ as follows:
$bc_t=\delta_{v_t\bullet}(r) \cdot |\mathcal{RV}(r)|$.
We have:
\begin{align*}
\mathbb E\left[bc_t\right] &= \sum_{v \in \mathcal{RV}(r)}
\left( \frac{1}{|\mathcal{RV}(r)|} \cdot \delta_{v_t\bullet}(r) \cdot |\mathcal{RV}(r)| \right) = BC(r).
\end{align*}
The random variable $bc$ is the average of $T$ independent random variables $bc_t$.
Therefore, we have:
\begin{align*}
\mathbb E\left[ bc \right] = \frac{\sum_{t=1}^T \mathbb E\left[ bc_t \right]}{T}
                 = \frac{T \cdot \mathbb E\left[ bc_t \right]}{T} = BC(r).
\end{align*}
\end{proof}

\begin{proposition}
\label{proposition:errorbound}
In Algorithm~\ref{algorithm:app-bcd},
let $K$ be the maximum dependency score that a vertex may have on $r$.
For a given $\epsilon \in \mathbb R^+$, we have:
\begin{equation}
\label{eq:errorbound}
\mathbb P\left[\left| BC(r)- bc \right| > \epsilon \right] \leq
2\exp \left( -2  T \cdot \left( \frac{\epsilon}{K\cdot |\mathcal{RV}(r)|} \right)^2  \right).
\end{equation}
\end{proposition}

\begin{proof}
The proof is done using Hoeffding's inequality \cite{Hoeffding:1963}.
Let $X_1, \ldots, X_n$ be independent random variables bounded by the interval
$[a, b]$, i.e., $a \leq X_i \leq b$ ($1 \leq i \leq n$).
Let also $\bar{X}=\frac{1}{n}\left( X_1 + \ldots + X_n \right)$.
Hoeffding \cite{Hoeffding:1963} showed that:
\begin{equation}
\label{eq:hoeffding}
\mathbb P\left[ \left| \mathbb E\left[ \bar{X} \right] - \bar{X} \right| > \epsilon \right] \leq
2 \exp \left( -2 n \cdot \left( \frac{\epsilon}{b-a} \right)^2 \right).
\end{equation}

Similar to the proof of Proposition~\ref{proposition:expectedvalue},
for each $t$, $ 1 \leq t \leq T$, we define random variable $bc_t$ as follows:
$bc_t=\delta_{v_t\bullet}(r) \cdot |\mathcal{RV}(r)|$.
Note that in Algorithm~\ref{algorithm:app-bcd} vertices $v_t$
are chosen independently, as a result, random variables $bc_t$ are independent, too.
Hence, we can use Hoeffding's inequality,
where
$X_i$'s are $bc_t$'s,
$\bar{X}$ is $bc$,
$n$ is $T$,
$a$ is $0$ and
$b$ is $K\cdot |\mathcal{RV}(r)|$.
Putting these values into Inequality~\ref{eq:hoeffding}
yields Inequality~\ref{eq:errorbound}.
\end{proof}

Inequality~\ref{eq:errorbound} says that for given
values $\epsilon \in \mathbb R^+$ and $\delta \in (0,1)$,
if $T$ is chosen such that
\begin{equation}
\label{eq:samplesnumber}
T \geq
\frac{\ln\left(\frac{2}{\delta} \right) \cdot K^2 \cdot |\mathcal{RV}(r)|^2}{2 {\epsilon}^2 },
\end{equation}
then, Algorithm~\ref{algorithm:app-bcd} estimates
betweenness score of $r$
within an additive error $\epsilon$ with a probability at least $1-\delta$.
The difference between Inequality~\ref{eq:samplesnumber} and the number of samples required by
the methods that uniformly sample from the set of all vertices (e.g., \cite{jrnl:Brandes3})
is that in the later case, the lower bound on the number of samples is a function
of $|V(G)|^2$, instead of $|\mathcal{RV}(r)|^2$.
As mentioned earlier, for most of the vertices, $|\mathcal{RV}(r)| \ll |V(G)|$.

\section{Experimental results}
\label{sec:experimentalresults}

We perform extensive experiments
on several real-world networks to assess
the quantitative and qualitative behavior of our proposed exact and approximate algorithms.
The experiments are done on an Intel
processor clocked at 2.6 GHz with 16 GB
main memory, running Ubuntu Linux 16.04 LTS.
All the programs are compiled by the GNU C++ compiler 5.4.0 using optimization level 3.

\medskip
We test the algorithms over several real-world datasets from different domains, including
the {\em amazon} product co-purchasing network \cite{DBLP:journals/tweb/LeskovecAH07},
the {\em com-dblp} co-authorship network \cite{DBLP:conf/icdm/YangL12},
the {\em com-amazon} network \cite{DBLP:conf/icdm/YangL12}
the {\em p2p-Gnutella31} peer-to-peer network \cite{DBLP:journals/tkdd/LeskovecKF07},
the {\em slashdot} technology-related news network \cite{DBLP:conf/chi/LeskovecHK10} and
the {\em soc-sign-epinions} who-trust-whom online social network \cite{DBLP:conf/chi/LeskovecHK10}.
All the networks are treated as directed graphs.
Table~\ref{table:dataset} summarizes specifications of our real-world networks.

\makesavenoteenv{tabular}
\makesavenoteenv{table}

\begin{table}[h]\small
\vspace*{-2mm}
\caption{\label{table:dataset}Summary of real-world datasets.}\vspace*{-1mm}
\begin{center}
\begin{tabular}{ l | l l  }
\hline
Dataset &\# vertices & \# edges \\ \hline
Amazon\footnote{\url{http://snap.stanford.edu/data/amazon0302.html}} & 262,111 & 1,234,877  \\
Com-amazon\footnote{\url{http://snap.stanford.edu/data/com-Amazon.html}}  & 334,863 & 925,872 \\
Com-dblp\footnote{\url{http://snap.stanford.edu/data/com-DBLP.html}} & 317,080 & 1,049,866  \\
Email-EuAll\footnote{\url{https://snap.stanford.edu/data/email-EuAll.html}} & 224,832 &340,795 \\
P2p-Gnutella31\footnote{\url{http://snap.stanford.edu/data/p2p-Gnutella31.html}} & 62,586 & 147,892  \\
Slashdot\footnote{\url{http://snap.stanford.edu/data/soc-sign-Slashdot090221.html}} & 82,144 & 549,202  \\
Soc-sign-epinions\footnote{\url{http://snap.stanford.edu/data/soc-sign-epinions.html}} &  131,828 & 841,372 \\
Web-NotreDame\footnote{\url{https://snap.stanford.edu/data/web-NotreDame.html}} & 325,729  & 1,497,134 \\
\hline
\end{tabular}
\end{center}\vspace*{-3mm}
\end{table}

As mentioned before,
for a directed graph $G$ and a vertex $r \in V(G)$,
both of our proposed exact and approximate algorithms
first compute $\mathcal{RV}(r)$, which can be done very effectively.
Then, based on the size of $\mathcal{RV}(r)$,
someone may decide to use either the exact algorithm or the approximate algorithm.
Hence in our experiments, we follow the following procedure:
\begin{itemize}
\item first, compute $\mathcal{RV}(r)$,
\item then, if $|\mathcal{RV}(r)| \leq \tau$, run \textsf{E-BCD};
otherwise, run \textsf{A-BCD} with $\tau$ as the number of samples.
\end{itemize}
We refer to this procedure as \textsf{BCD}.
The value of $\tau$ depends on the amount of time someone wants to spend for computing betweenness centrality.
In our experiments reported here, we set $\tau$ to 1000.
We compare our method against the most efficient existing algorithm for approximating betweenness centrality,
which is \textsf{KADABRA} \cite{DBLP:conf/esa/BorassiN16}.

\medskip
For a vertex $r \in V(G)$, its empirical approximation error is defined as:
\begin{equation}
\label{eq:vertexerror}
Error(v)=\frac{|App(v)-BC(v)| }{BC(v)} \times 100,
\end{equation}
where $App(v)$ is the calculated approximate score.

\subsection{Results}

Table~\ref{table:1000random} reports the results of our first set of experiments.
For \textsf{KADABRA}, we have set $\epsilon$ and $\delta$ to $0.01$ and $0.1$, respectively\footnote{For given
values of $\epsilon$ and $\delta$,
\textsf{KADABRA} computes the {\em normalized betweenness}
of the vertices of the graph within an error $\epsilon$
with a probability at least $1-\delta$.
The {\em normalized betweenness} of a vertex is its betweenness score divided by
$|V(G)|\cdot\left(|V(G)|-1\right)$.
Therefore, we multiply the scores computed by \textsf{KADABRA} by
$|V(G)|\cdot\left(|V(G)|-1\right)$.
}.
Then from each dataset we choose $1000$
vertices uniformly at random and run \textsf{BCD} for any of these
vertices.
For the \textsf{BCD} algorithm, we report
both "Avg. time"
and "$\text{Avg. time}_{\mathcal{RV}}$",
where "$\text{Avg. time}_{\mathcal{RV}}$" is the average of run times of computing $\mathcal{RV}$
and "Avg. time" is the average of run times of the other parts of the algorithm.
The average total running time of \textsf{BCD}
is the sum of "Avg. time" and "$\text{Avg. time}_{\mathcal{RV}}$".
We also report
"Avg. error", which is the average of
empirical approximation errors
(defined in Equation~\ref{eq:vertexerror}),
and "$\%$exact" that presents the percentage of the vertices
for which \textsf{BCD}
computes betweenness scores exactly, hence, their approximation error is $0$.
We remind that if $|\mathcal{RV}| \geq 1000$, approximate \textsf{BCD} is used, otherwise, exact \textsf{BCD} is employed.
As can be seen in the table, \textsf{BCD} estimates
betweenness centrality of a single vertex much faster
and with much less error.
It is notable that in most cases,
\textsf{BCD} computes the exact score within a tiny time,
whereas \textsf{KADABRA}
estimates the score with a large error within a much longer time.

\begin{table*}[!h]
\vspace*{-2mm}
\caption{\label{table:1000random}Empirical evaluation of \textsf{BCD} against \textsf{KADABRA}
for $1000$ randomly chosen vertices.
Values of $\delta$ and $\epsilon$ are $0.1$ and $0.01$, respectively.
All the reported times are in seconds.
The number of samples in \textsf{A-BCD} is $1000$.
'$\%$exact' presents the percentage of the vertices for which betweenness scores are computed exactly by \textsf{BCD-E} and 'Avg. $\text{time}_{\mathcal{RV}}$' presents the average time to compute reachable vertices.
}\vspace*{-2mm}
\begin{center}
\resizebox{\textwidth}{!}{%
\begin{tabular}{ l | l l | l l l | l l l l}
\hline
Dataset & \multicolumn{2}{|c|}{Randomly chosen vertices} & \multicolumn{3}{|c|}{\textsf{KADABRA}} & \multicolumn{4}{|c}{\textsf{BCD}} \\
       & Avg. $|\mathcal{RV}(r)|$ & Avg. $\frac{|\mathcal{RV}(r)|}{|V(G)|}$ &
           $\#$samples & Time & Error ($\%$) & $\%$exact & Avg. time & Avg. $\text{time}_{\mathcal{RV}}$ & Avg. error ($\%$)\\
\hline
Amazon   & 3453.714 & 0.013 & 16739 & 19.14  & 100 & 92.857 & 0.673 & 0.331 & 0.018 \\
Com-amazon & 74.533  & 0.0002 & 15036 & 27.70 & 100  & 100 & 0.512 & 0.367 & 0 \\
Com-dblp & 24635.923 & 0.077 & 17873 & 26.14 & 100 & 69.230 & 2.14 & 0.322 & 2.678 \\
Email-EuAll & 13652.785 & 0.0607 & 17066 & 16.01 &  100 & 64.285 & 0.964 &  0.083 &0.995 \\
P2p-Gnutella31  & 7246.071 & 0.115 & 16401 & 6.88 & 100 & 57.142 & 2.221 & 0.046 & 5.854 \\
Slashdot0902  & 6662.866 & 0.0811 & 17421 & 7.95 & 100 & 80 & 0.995 & 0.130 & 6.279 \\
Soc-sign-epinions & 14567.875 & 0.110 & 19099 & 11.28 & 100  & 62.5 & 1.789 & 0.150  & 9.234 \\
Web-NotreDame  & 431.714 & 0.001 & 19908 & 27.29 & 100  & 85.714  & 0.852 & 0.240  &  0.041  \\
\hline
\end{tabular}
}
\end{center}\vspace*{-2mm}
\end{table*}

\begin{table*}[!ht]
\caption{\label{table:randomexp1}Empirical evaluation of \textsf{BCD} against \textsf{KADABRA}
for some of randomly chosen vertices.
Values of $\delta$ and $\epsilon$ are $0.1$ and $0.01$, respectively.
All the reported times are in seconds.
The number of samples in \textsf{A-BCD} is $1000$.
}\vspace*{-7mm}
\begin{center}
\resizebox{\textwidth}{!}{%
\begin{tabular}{ l | l l l l | l l l | l l l l}
\hline\hline
Dataset & \multicolumn{4}{|c|}{Randomly chosen vertices} & \multicolumn{3}{|c|}{\textsf{KADABRA}} & \multicolumn{4}{|c}{\textsf{BCD}} \\
        & $r$ & $BC(r)$ & $|\mathcal{RV}(r)|$ & $\frac{|\mathcal{RV}(r)|}{|V(G)|}$ &
           $\#$samples & Time & Error ($\%$) & E/A & Time & $\text{Time}_{\mathcal{RV}}$ & Error ($\%$)\\
\hline\hline
Amazon  & 13645  & 19613.1 & 47187 & 0.1800& 16739 & 19.14  & 100 & A & 2.60 & 0.26 & 0.26 \\
        & 91289  & 87523.6 & 150  &0.0005 &  &   & 100 & E & 0.67 & 0.29 & 0  \\
        & 17054  & 35752.6 & 533   &0.0020 &  &  & 100 & E & 1.26 & 0.29 &  0 \\
        & 231249 & 10449.4 & 4  &0.00001 &  &    & 100 &  E & 0.11 & 0.30 & 0 \\
        & 246486 & 1837.58 & 34  &0.0001 &  &    & 100 & E & 0.17 & 0.30 & 0 \\
\hline
Com-amazon& 202389 & 1486.8 & 13 & 0.00003 & 15036 & 27.70 & 100  & E & 0.14 & 0.27 & 0    \\
          & 263212 & 364 & 3 & 0.000008 & &  &  100 & E & 0.12 & 0.27 & 0\\
          & 81097 & 11 & 14 & 0.00004 & &  &  100 & E & 0.15 & 0.27 & 0 \\
          & 13732 & 1701.51 & 616 & 0.0018 & &  & 100 & E & 1.41 &  0.28 & 0 \\
          & 29825 & 139 & 15 & 0.00004 & &  &  100 & E & 0.15 & 0.27 & 0 \\
\hline
Com-dblp  & 4456   & 10153 & 2092 & 0.0065& 17873 & 26.14 & 100 &  A & 5.74 & 0.26 & 1.10 \\
          & 278950 & 34326.5 & 11 & 0.00003  & &  & 100 & E & 0.13 & 0.27 & 0   \\
          & 244680 & 232994 & 22 & 0.00006  & &  & 100 & E &  0.21 & 0.27 & 0 \\
          & 21141  & 1957.93 & 73 & 0.0002 & &  & 100 & E &  0.48 & 0.27 & 0 \\
          & 129908 & 303543 & 41 & 0.0001 & &  &  100 & E & 0.53 & 0.29 & 0 \\
\hline
Email-EuAll& 25362  & 1869.16 & 2 & 0.000008 & 17066 & 16.01 &  100 & E & 0.03 & 0.08 & 0\\
           & 16682  & 2269.29 & 64 & 0.0002 & &  &  100 & E & 0.14 & 0.08 & 0 \\
           & 8796   & 241434 & 21181 & 0.0942 & &  & 100 & A & 1.88 &  0.07 & 1.72 \\
           & 50365  & 3 & 2 &  0.000008   & &  &  100 &  E & 0.03 & 0.07 & 0 \\
           & 2139   & 503650 & 111674 & 0.4966 & &  &  100 &  A & 1.78 & 0.08 & 3.59 \\
\hline
P2p-Gnutella31& 46263  & 12655.2 & 2 & 0.00003 & 16401 & 6.88 & 100  & E & 0.03 & 0.04 & 0\\
              & 34547  & 3538.79 & 173 & 0.0027 & &   &  100 &  E & 0.95 & 0.04 & 0\\
              & 54609  & 27824.9 & 3  & 0.00004 & &   &  100 &  E & 0.03 & 0.04 &  0 \\
              & 37518  & 6175.2 & 24141 & 0.3857 & &  &    100 &  A & 2.44 & 0.06 & 11.31 \\
              & 9781   & 4582130 &   3 & 0.00004& &   &  100 &  E & 0.02 & 0.04 & 0 \\
\hline
Slashdot0902  & 20825  & 15940.9 & 21 & 0.0002& 17421 & 7.95 & 100  & E & 0.17 & 0.16 & 0 \\
              & 47806  & 15891.7 & 3 &0.00003 & &    &  100 & E & 0.06 & 0.15 & 0 \\
              & 48251  & 21744 & 3 & 0.00003& &   &  100 & E &  0.05 & 0.15 & 0 \\
              & 20969  & 43067 & 369 &0.0044 & &   &  100 &  E & 2.30 & 0.17 & 0 \\
              & 57099  & 6165.01 & 2 & 0.00002& &   &  100 &  E &  0.05 & 0.15 & 0 \\
\hline
Soc-sign-epinions & 2740   & 2352.43 & 36393 & 0.2760& 19099 & 11.28 & 100  & A & 4.57 & 0.17 & 55.34 \\
                  & 24080  & 9198.78 & 2621 &0.0198  & &    &  100 & A & 4.60 & 0.15 & 18.48\\
                  & 38349  & 75201.9 & 35 & 0.0002 & &     &  100 & E & 0.24 & 0.14 & 0 \\
                  & 82156  & 8802 & 34 & 0.0002  & &    &   100 & E & 0.19 & 0.14 & 0 \\
                  & 38266  & 8052 & 3 & 0.00002 & &   &  100 &  E & 0.04 & 0.14 & 0 \\
\hline
Web-NotreDame     & 21026   & 140 & 9 &  0.00002 & 19908 & 27.29 & 100  & E & 0.08 & 0.25  & 0  \\
                  & 133847  & 9003.53 & 797 &0.0024 & &    &  100 &  E & 1.84 & 0.25  & 0 \\
                  & 307622  & 4212.33 & 44 & 0.0001 & &    &   100 & E & 0.18 & 0.25 & 0  \\
                  & 176211  & 2157.42 & 30 & 0.00009  & &   &   100 & E & 0.14 & 0.25 & 0  \\
                  & 307134  & 3079.5 & 123 & 0.0003 & &    &   100 & E & 0.35 & 0.25 & 0 \\
\hline\hline
\end{tabular}
}
\end{center}\vspace*{-3mm}
\end{table*}

In order to investigate the behavior of the algorithms
more deeply, over each dataset we choose $5$ vertices at random and
report their results in Table~\ref{table:randomexp1}.
This table has a column, called "A/E", where "E" means that
the computed score by \textsf{BCD} is
exact (hence, the approximation error is $0$) and
"A" means that $\mathcal{RV}$ is larger than $1000$, therefore
approximate \textsf{BCD} has been employed.

As can be seen in Table~\ref{table:randomexp1}, for most of the randomly picked up vertices,
$\mathcal{RV}$ is very small and it can be computed very efficiently.
This gives exact results in a very short time, less than $3$ seconds in total.
In all these cases, while \textsf{KADABRA} spends considerably more time,
since it estimates that the normalized betweenness scores are less than the error bound $\epsilon$,
it simply estimates them as $0$.\footnote{\textsf{KADABRA}
aims to provide an estimation whose error,
with a high probability, is at most $\epsilon$.
Therefore, when it estimates that with a high probability the betweenness score of a
vertex is less than $\epsilon$, it estimates the score as $0$.
In this way, with high probability the theoretical error will be bounded by $\epsilon$.}
Therefore, its empirical approximation error
becomes $100\%$.
The randomly picked up vertices belong to the different ranges of betweenness scores,
including high, medium and low.

\medskip
After observing these experimental results,
someone may be interested in the following questions:
\begin{itemize}
\item[\textbf{Q1.}]
The accuracy of \textsf{KADABRA} depends on the values of $\epsilon$ and $\delta$.
Can changing (increasing or decreasing) their values improve
the performance of \textsf{KADABRA} and make it be comparable to \textsf{BCD}?
\item[\textbf{Q2.}]
\textsf{KADABRA} is more efficient for the vertices that have the highest betweenness scores
and since most of the randomly chosen vertices do not have a very high betweenness score,
compared to \textsf{EBC}, \textsf{KADABRA} does not show a good performance.
What is the efficiency of \textsf{BCD}, compared to \textsf{KADABRA}, for the vertices
that have the highest betweenness scores?
\item[\textbf{Q3.}]
In the experiments reported in Table~\ref{table:randomexp1},
\textsf{BCD} is used to estimate betweenness score of only one vertex.
However, in practice it might be required to estimate betweenness scores of
a given set of vertices.
How efficient is \textsf{BCD} in this setting?
\end{itemize}

In the rest of this section, we answer these questions.

\begin{table}[!h]
\vspace*{-1mm}
\caption{\label{table:randomexp2}Empirical evaluation of \textsf{KADABRA}
for $\delta=0.1$, $\epsilon=0.005$ and $0.05$.}
\tiny
\begin{center}
\begin{tabular}{ l | l | l l l | l l l }
\hline\hline
Dataset & Vertex $r$ & \multicolumn{3}{|c|}{\textsf{KADABRA} ($\epsilon=0.005$)} & \multicolumn{3}{|c}{\textsf{KADABRA} ($\epsilon=0.05$)} \\
        &   & $\#$samples & Time &  Error ($\%$) & $\#$samples & Time & Error ($\%$) \\
\hline\hline
Amazon  & 13645  & 47330 & 53.98  & 100  & 1615 & 3.648 & 100 \\
        & 91289  &  &  &  100  &  &  & 100 \\
        & 17054  &  &  &  100  &  &  & 100 \\
        & 231249 &  &  &  100   &  &  & 100 \\
        & 246486 &  &  &  100  &  &  & 100 \\
\hline
Com-amazon& 202389 & 42207 & 58.76 & 100  &  1390 & 4.36 & 100 \\
          & 263212 &  &  &  100   &    &  & 100 \\
          & 81097 &   &  &  100   &   &   & 100 \\
          & 13732 &    &  &  100   &   &   & 100 \\
          & 29825 &   &  &  100   &    &  & 100 \\
\hline
Com-dblp  & 4456   &  50667  & 77.40 & 100  & 1627 & 4.15 & 100 \\
          & 278950 &   &  &  100   &  &  & 100 \\
          & 244680 &    &  & 100   &  &  & 100 \\
          & 21141  &    &  & 100   &  &  & 100 \\
          & 129908 &    &  & 100  &  &  & 100 \\
\hline
Email-EuAll& 25362  & 48079 & 43.43 &  100  &  1390 & 2.28 & 100 \\
           & 16682  &    &  &  100  &  &  & 100 \\
           & 8796   &    &  &  100  &  &  & 100 \\
           & 50365  &    &  &  100  &  &  & 100 \\
           & 2139   &   &  &   100  &  &  & 100 \\
\hline
P2p-Gnutella31& 46263  &  47631  & 18.12 &  100  & 1445 & 0.81 & 100 \\
              & 34547  &    &  &  100  &  &  & 100 \\
              & 54609  &    &  &  568.32  &  &  & 100 \\
              & 37518  &    &  & 100  &  &  & 100 \\
              & 9781   &    &  &  6.52  &  &  & 100 \\
\hline
Slashdot0902  & 20825  &  50776  & 22.38 & 100  & 1542 & 0.94 & 100 \\
              & 47806  &    &  &  100  &  &  & 100 \\
              & 48251  &    &  & 100  &  &  & 100 \\
              & 20969  &    &  &  100  &  &  & 100 \\
              & 57099  &    &  &  100  &  &  & 100 \\
\hline
Soc-sign-epinions & 2740   &  53667  & 30.54 &  100  & 1479 & 1.90 & 100 \\
                  & 24080  &    &  &  100  &  &  & 100 \\
                  & 38349  &    &  &  100  &  &  & 100 \\
                  & 82156  &    &  &  100  &  &  & 100 \\
                  & 38266  &    &  &  100  &  &  & 100 \\
\hline
Web-NotreDame     & 21026   &  51015  & 73.92 & 100  & 1935 & 2.21 & 100 \\
                  & 133847  &    &  &  100  &  &  & 100 \\
                  & 307622  &    &  &  100  &  &  & 100 \\
                  & 176211  &    &  & 100  &  &  & 100 \\
                  & 307134  &    &  & 100  &  &  & 100 \\
\hline\hline
\end{tabular}
\end{center}\vspace*{-1mm}
\end{table}

\begin{table}[htbp]
\caption{\label{table:highestexp}Empirical evaluation of \textsf{BCD} against \textsf{KADABRA-TOP-1}
for vertices with the highest betweenness scores.
The value of $\delta$ is 0.1.
All the reported times are in seconds. The number of samples in \textsf{A-BCD} is $1000$.
}
\tiny
\begin{center}
\rotatebox{90}{
\begin{tabular}{ l | l l l l | l l l l | l l l}
\hline \hline
Dataset & \multicolumn{4}{|c|}{Vertex with the highest BC} & \multicolumn{4}{|c|}{\textsf{KADABRA-TOP-1}} & \multicolumn{3}{|c}{\textsf{BCD}} \\
        & $r$ & $BC(r)$ & $|\mathcal{RV}(r)|$ & $\frac{|\mathcal{RV}(r)|}{|V(G)|}$ &
           $\epsilon$ & $\#$samples & Time & Error ($\%$) & Time & $\text{Time}_{\mathcal{RV}}$ & Error ($\%$)\\
\hline\hline
Amazon  & 2804  & 16066000 & 162707 &0.6207 &  0.01 & 16181 & 0.26 & - & 2.38 & 0.29  & 1.35  \\
        &       &         &       & &  0.005 &  45320  &  0.56   &  71.69 &    &   &   \\
        &       &         &       & &  0.0005 &  1459502    &  16.65  & 3.01 &    &   &   \\
\hline
Com-amazon& 28081 & 378550 & 3812 & 0.0113 & 0.01 & 14619 & 0.14  & - & 2.31 & 0.28 & 0.52 \\
        &       &         &       & &  0.005 &  40590   &   0.21    &  -  &    &   &   \\
        &       &         &       & &  0.0005 &  1249908    &  3.86 & 28.90 &    &   &   \\
\hline
Com-dblp  & 49124   & 24821300 & 70561 & 0.2225 & 0.01 & 17303 & 0.64 &  17.04 & 6.27 & 0.27 & 9.77 \\
        &       &         &       & &  0.005 &  48411   &    1.62   &  7.96  &    &   &   \\
        &       &         &       & &  0.0005 &    1581635  &  54.11   & 6.79  &    &   &   \\
\hline
Email-EuAll& 2387  & 15943100 & 102596 & 0.4563 & 0.01 & 16588 &  0.10 & 33.79 & 1.76 & 0.08 & 3.37 \\
        &       &         &       & &  0.005 &  46123  &  0.17  &  17.50  &    &   &   \\
        &       &         &       & &  0.0005 &   1471932   &   3.87  & 4.04  &    &   &   \\
\hline
P2p-Gnutella31& 9781  & 4580850 & 36141 & 0.5774 & 0.01 & 13618 & 0.32  & 57.61 & 1.78 & 0.04 & 2.59 \\
        &       &         &       & &   0.005 &  40909  &  1.00   & 6.51   &    &   &   \\
        &       &         &       & &  0.0005 &   1515822 & 38.31  & 0.32  &    &   &   \\
\hline
Slashdot0902  & 18238  & 8531850 & 19153 & 0.2331 & 0.01 & 16962 & 0.99  & 11.96 & 3.90 & 0.10 & 3.37 \\
        &       &         &       & &   0.005 &  44847  &  2.52     &  5.87  &    &   &   \\
        &       &         &       & &  0.0005 & 1718486 &  103.89 &  0.16  &    &   &   \\
\hline
Soc-sign-epinions & 27463   & 26116100 & 9880 &0.0749 & 0.01 & 18601 & 1.10  & 2.25 & 5.43 & 0.12 & 2.30 \\
        &       &         &       & &  0.005 &   51502  &  2.97     &  0.23  &    &   &   \\
        &       &         &       & &  0.0005 &   2398143   &  143.92   &   1.61  &    &   &    \\
\hline
Web-NotreDame     & 7137   & 323101000 & 233965 & 0.7182  & 0.01 & 19448 & 0.18  & 1.30 & 2.71 & 0.235  & 0.26  \\
        &       &         &       & &  0.005 &  49456   &    0.30   &   7.56  &    &   &   \\
        &       &         &       & &  0.0005 &  779273  &  3.93  &  2.22  &    &   &   \\
\hline\hline
\end{tabular}
}
\end{center}
\end{table}

\paragraph{\textbf{Q1}}
To answer Q1, first we fix $\delta$ to $0.1$ and
run \textsf{KADABRA} with $\epsilon=0.005$ (i.e., with a lower value)
and $\epsilon=0.05$ (i.e., with a higher value).
The results are reported in Table~\ref{table:randomexp2}.
In these two settings,
most of the scores estimated by \textsf{KADABRA} are still $0$.
There are only two exceptions where, however,
the approximation error is high.
For $\epsilon=0.005$,
the running time of \textsf{KADABRA}
is considerably more than its running time for $\epsilon=0.01$ and as a result,
the running time of \textsf{BCD}.
However, paying this extra cost does not improve its accuracy,
with respect to \textsf{BCD}.
Increasing $\epsilon$ to $0.05$, reduces running time of \textsf{KADABRA}
and makes it comparable to the running time of \textsf{BCD}.
However, \textsf{BCD} shows a much better accuracy.

\medskip
Then, we fix $\epsilon$ to $0.01$ and
run \textsf{KADABRA} with $\delta=0.05$ (i.e., with a lower value)
and $\delta=0.15$ (i.e., with a higher value).
In these cases, we do not observe meaningful changes in the behavior
(running time and accuracy) of \textsf{KADABRA}.
We may only state that in the case of $\delta=0.15$,
the algorithm works slightly faster.
As a result, it seems \textsf{KADABRA} is less sensitive to
the value of $\delta$ than to the value of $\epsilon$.
Due to the high similarity of the results obtained in these two cases to
the results of Table~\ref{table:randomexp1}, we do not report them.

\paragraph{\textbf{Q2}}
To answer Q2, over each dataset we examine the algorithms
for the vertex that has the highest betweenness score\footnote{We already
find this vertex using the exact algorithm.}.
The results are reported in Table~\ref{table:highestexp}.
\textsf{KADABRA} can be optimized to estimate
betweenness centrality of only top $k$ vertices, where $k$ is an input parameter.
In the experiments of this part, we use this optimized version of \textsf{KADABRA} with $k=1$
and refer to it as \textsf{KADABRA-TOP-1}.
In \textsf{KADABRA-TOP-1},
we consider three values for $\epsilon$: $0.01$, $0.005$ and $0.0005$
and in all the cases, we set $\delta$ to $0.1$.
Similar to the other experiments, we run \textsf{BCD} with $\tau=1000$.
In all the experiments of this part, the size of $\mathcal{RV}$ becomes larger than $1000$, hence,
the scores computed by \textsf{BCD} are approximate scores.
In Table~\ref{table:highestexp}, in three cases the error of \textsf{KADABRA-TOP-1} is not reported.
The reason is that in these cases the vertex $r$ that has the highest betweenness score,
is not among the vertices considered by \textsf{KADABRA-TOP-1} as a top-score vertex.
Hence, \textsf{KADABRA-TOP-1} does not report any value for it.

\medskip
In this setting, none of the
algorithms outperforms the other one in all the cases.
More precisely, while for some values of $\epsilon$
\textsf{KADABRA-TOP-1} has a better accuracy as well as a higher running time,
in some other cases the story is in the other way.
Nevertheless, we can
investigate the datasets one by one.
Over {\em amazon}, for all values of $\epsilon$,
\textsf{BCD} has a better approximation error than \textsf{KADABRA-TOP-1}.
In particular, for $\epsilon=0.0005$,
\textsf{KADABRA-TOP-1} takes much more time but produces a less accurate output.
Hence, we can argue that over {\em amazon} \textsf{BCD} outperforms \textsf{KADABRA-TOP-1}.
The same holds for {\em com-amazon}, {\em email-EuAll} and {\em web-NotreDame}
and over all these datasets, \textsf{BCD} outperforms \textsf{KADABRA-TOP-1}.
Over {\em com-dblp},
for $\epsilon=0.005$, \textsf{KADABRA-TOP-1} outperforms \textsf{BCD} in terms of both
accuracy and running time.
This also happens over {\em soc-sign-epinions} for $\epsilon=0.01$ and $0.005$.
Hence, someone may argue that over these two datasets \textsf{KADABRA-TOP-1} outperforms \textsf{BCD}.
Over {\em p2p-Gnutella31} and {\em slashdot0902},
on the one hand
for $\epsilon=0.01$ and $0.005$,
\textsf{BCD} shows a better accuracy, however, it is slightly slower.
On the other hand, for $\epsilon=0.0005$, \textsf{KADABRA-TOP-1} shows a better accuracy,
however, it takes much more time.
Altogether, we can say that for estimating betweenness scores of the vertices
that have the highest scores, in most of the datasets \textsf{BCD} works better than \textsf{KADABRA-TOP-1}.

\begin{table*}[!htb]
\caption{\label{table:setexp}Empirical evaluation of
estimating betweenness scores of a set of vertices.
All the reported times are in seconds. The number of samples in \textsf{A-BCD} is $1000$.
}
\begin{center}
\resizebox{\textwidth}{!}{%
\begin{tabular}{ l | l | l l l | l l | l l l }
\hline
Dataset & Set size & \multicolumn{3}{|c|}{Error ($\%$)} & Time & $\text{Time}_{\mathcal{RV}}$ & \multicolumn{3}{|c}{$\mathcal{RV}$ size} \\
        &    &  Avg.  &  Max. & Min. &  &  &  Avg.  &  Max. & Min. \\
\hline
Amazon  & 5   &  1.47 & 7.10 &  0  & 4.81  & 1.44 &  9581.6 & 47187 & 4  \\
        & 10  &  0.73 & 7.10 &  0  &  7.42 & 3.21 & 4818.4 & 47187 & 1   \\
        & 15  &  0.88 & 7.10 &  0  &  9.74 & 4.98 & 3497.798 & 47187 & 1   \\
\hline
Com-amazon& 5   &  0  & 0  &  0  &  1.98 & 1.36 & 132.2 & 616 & 3  \\
        & 10  &   0  &  0  &  0  & 4.92 & 3.43 & 91.2 & 616 & 2   \\
        & 15  &  0  &  0  &  0  &  7.07 &  5.48 & 65.93 & 616 & 1   \\
\hline
Com-dblp& 5   &  0.22  & 1.10  &  0  &  7.09 & 1.36 & 447.8 & 2092 & 11   \\
        & 10  & 3.47 & 19.45 &  0  & 20.71 & 3.08  & 24483.6 & 227218 & 1  \\
        & 15  & 2.32 & 19.45 &  0  & 28.81 & 4.92  & 21351.33 & 227218 & 1  \\
\hline
Email-EuAll& 5   &  1.06  &  3.59  &  0  & 3.86  & 0.38  & 26584.6 & 111674 & 2  \\
        & 10  &  1.39  & 7.95 &  0  & 9.76 & 0.78  & 19020.9 & 111674 & 2  \\
        & 15  &  0.93  &  7.95  &  0  & 13.52 & 1.27  & 12742.8 & 111674 &  2  \\
\hline
P2p-Gnutella31 & 5   &   2.26  &  11.31   &  0  & 3.47  &  0.22   & 4864.2 & 24141 & 2  \\
               & 10  &   7.26   &  39.17   &  0  &  23.09 &  0.46  & 5493.6 & 24141 & 2   \\
               & 15  &  6.79   &  39.17   &  0  &  33.27 &  0.72  & 8637.73 & 28122 & 2   \\
\hline
Slashdot0902  & 5   &    0  &  0   &  0  &  2.62  &  0.78   & 79.6 & 369 & 2  \\
        & 10  &   5.04   &  50.48   &  0  &  11.37  &  1.38   & 3784.3 & 26802 & 1    \\
        & 15  &   4.92   &  50.48   &  0  &  14.93  &  1.99   & 6662.86 & 62089 & 1    \\
\hline
Soc-sign-epinions & 5   &  13.37  &  48.37 &  0  & 9.64   &  0.74  & 7817.2 & 36393 & 3 \\
        & 10  &  9.68  &  48.37   &  0  &  17.71  &  1.52  & 20302.7 & 109520 & 1  \\
        & 15  &  9.38  &  48.37  &  0  &  28.46  &  2.28  & 15538.86 & 109520 & 1   \\
\hline
Web-NotreDame   & 5   &   0   & 0    &  0  & 2.58   & 1.25  & 200.6 & 797 & 9 \\
        & 10  &   0   &  0   &  0  &  6.89  & 2.44   & 231.5 & 1092 & 9  \\
        & 15  &   0.03   &  0.30   &  0  &  13.16  &  3.62  & 414.46 & 2610 & 1  \\
\hline
\end{tabular}
}
\end{center}
\end{table*}

\paragraph{\textbf{Q3}}
To answer Q3,
we select a random set of vertices and
run \textsf{BCD} for each vertex in the set.
The results are reported in Table~\ref{table:setexp}, where the set contains 5, 10 or 15 vertices.
Over all the datasets and for each set of vertices, we report the
average, maximum and minimum errors of the vertices.
For all the datasets, minimum error is always 0.
In Table~\ref{table:setexp},
"$\text{Time}_{\mathcal{RV}}$" is the total time of computing
$\mathcal{RV}$ of all the vertices in the set and "Time" is the total time of
the other steps of computing betweenness scores of all the vertices in the set.
Therefore, the total running time of \textsf{BCD} for a given dataset and a given set
is the sum of "Time" and "$\text{Time}_{\mathcal{RV}}$".
Comparing the results presented in Table~\ref{table:setexp} with the results presented
in Table~\ref{table:randomexp2} reveals that for estimating betweenness scores of a set of vertices,
\textsf{BCD} considerably outperforms \textsf{KADABRA} (where $\epsilon$ is $0.005$).
While in most cases the total running time of \textsf{BCD} is less than
the running time of \textsf{KADABRA} (even when the size of the set is 15),
\textsf{BCD} gives much more accurate results.
Note that even when in \textsf{KADABRA} $\epsilon$ is set to 0.01,
in many cases \textsf{BCD} is faster than \textsf{KADABRA}.
In particular, over datasets such as {\em amazon}, {\em com-amazon},
{\em email-EuAll} and {\em web-NotreDame},
even for the sets of size 15, \textsf{BCD} is faster than \textsf{KADABRA}
and it always produces much more accurate results.

\subsection{Discussion}

Our extensive experiments reveal that \textsf{BCD} usually significantly outperforms \textsf{KADABRA}.
This is due to the huge pruning that $\mathcal{RV}$ applies to the set of source vertices
that are used to form SPDs and compute dependency scores.
Note that in all the cases, $\mathcal{RV}$ is computed very efficiently, hence,
it does not impose a considerable load on the algorithm.
In the case of estimating betweenness score of the vertex with the highest betweenness score,
over two datasets we may argue that \textsf{KADABRA} outperforms \textsf{BCD}.
This has two reasons. On the one hand, in these cases the ratio $\frac{\mathcal{RV}(r)}{|V(G)|}$ is large,
as a result, many SPDs are computed by \textsf{BCD}.
On the other hand, the SPDs contain many vertices of the graph, as a result,
their computation is expensive.

\begin{figure}[!h]
\vspace*{-2mm}
\centering
\subfigure[]
{
\includegraphics[scale=0.5]{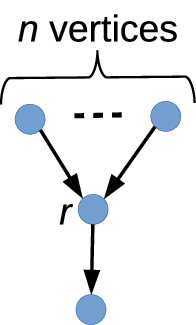}
\label{fig:discussion1}
}\qquad\qquad
\subfigure[]
{
\includegraphics[scale=0.5]{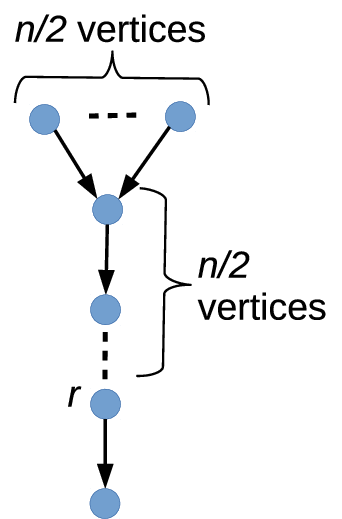}
\label{fig:discussion2}
}\vspace*{-2mm}
\caption
{
\label{fig:discussion}
Using \textsf{BCD},
in the graph of Figure~\ref{fig:discussion1}, $BC(r)$ is computed in $O(n)$ time;
whereas in the graph of Figure~\ref{fig:discussion2}, $BC(r)$ is computed in $O(n^2)$ time.
}
\end{figure}

In the end,
it is worth mentioning that while the size of $\mathcal{RV}$ is an important factor
on the efficiency of our algorithm, it is not the sole factor.
For example, both graphs of Figure~\ref{fig:discussion} have $n+2$ vertices,
the size of $\mathcal{RV}(r)$ in Figure~\ref{fig:discussion1} is $n$ and
the size of $\mathcal{RV}(r)$ in Figure~\ref{fig:discussion2} is $\frac{n}{2}$.
However, in Figure~\ref{fig:discussion1} each SPD is computed and processed in $O(1)$ time,
whereas in Figure~\ref{fig:discussion2} each SPD is computed and processed in $O(n)$ time.
Therefore, while in Figure~\ref{fig:discussion1} $BC(r)$ is computed in $O(n)$ time,
in Figure~\ref{fig:discussion2} it is computed in $O(n^2)$ time.

\section{Conclusion}
\label{sec:conclusion}

In this paper,
we studied the problem of computing betweenness score in large directed graphs.
First, given a directed network $G$ and a vertex $r \in V(G)$,
we proposed an exact algorithm to compute betweenness score of $r$.
Our algorithm first computes a set $\mathcal{RV}(r)$,
which is used to prune a huge amount of computations that do not contribute to the betweenness score of $r$.
Time complexity of our exact algorithm  is
respectively
$\Theta(|\mathcal{RV}(r)|\cdot|E(G)|)$ and
$\Theta(|\mathcal{RV}(r)|\cdot|E(G)|+|\mathcal{RV}(r)|\cdot|V(G)|\log |V(G)|)$
for unweighted graphs and weighted graphs with positive weights.
Then, for the cases where $\mathcal{RV}(r)$ is large,
we presented a simple randomized algorithm
that samples from $\mathcal{RV}(r)$ and performs computations for only
the sampled elements.
Finally, we performed extensive experiments over several real-world datasets from different domains
for several randomly chosen vertices as well as for the vertices
with the highest betweenness scores.
Our experiments revealed that for estimating betweenness score of a single vertex,
our algorithm considerably outperforms the most efficient existing randomized algorithms, in terms of both running time and accuracy.
They also showed that
our algorithm
improves the existing algorithms when someone
is interested in computing betweenness values of the vertices in a set
whose cardinality is very small ($15$ for the analyzed graphs).


\subsection*{Acknowledgement}
This work has been supported in part by the ANR project IDOLE.



\end{document}